\documentclass{aims}
\usepackage{amsmath,amssymb}
\usepackage{amsfonts}
  \usepackage{paralist}
  \usepackage{graphics} 
  \usepackage{epsfig} 
 \usepackage[colorlinks=true]{hyperref}

 \def\dref#1{(\ref{#1})}

\hypersetup{urlcolor=blue, citecolor=red}

\def\currentvolume{X}
 \def\currentissue{X}
  \def\currentyear{200X}
   \def\currentmonth{XX}
    \def\ppages{X--XX}

\newtheorem{theorem}{Theorem}[section]
\newtheorem{corollary}{Corollary}

\newtheorem{lemma}[theorem]{Lemma}
\newtheorem{proposition}{Proposition}

\theoremstyle{definition}
\newtheorem{definition}[theorem]{Definition}
\newtheorem{remark}{Remark}

\title[Consensus of Discrete-Time Linear Multi-Agent
Systems]
      {Consensus of Discrete-Time Linear Multi-Agent
Systems with Observer-Type Protocols}

\author[Zhongkui Li, Zhisheng Duan and Guanrong Chen]{}


 \email{zhongkli@gmail.com}
 \email{duanzs@pku.edu.cn}
 \email{eegchen@cityu.edu.hk}

\thanks{This work was supported by the National Science Foundation of China
under Grants 60974078 and 10832006.}

\begin{document}
\maketitle


\centerline{\scshape Zhongkui Li}
\medskip
{\footnotesize
 \centerline{School of Automation, Beijing Institute of Technology}
   \centerline{Beijing 100081, P. R. China}
}
\medskip

\centerline{\scshape Zhisheng Duan }
\medskip
{\footnotesize
 \centerline{State Key Lab for Turbulence and Complex Systems}
   \centerline{Department of Mechanics and Aerospace Engineering, College of
Engineering}
   \centerline{ Peking University,\, Beijing  100871, P. R. China}
} 

\medskip

\centerline{\scshape Guanrong Chen}
\medskip
{\footnotesize
 \centerline{Department of Electronic Engineering}
   \centerline{City University
of Hong Kong, Hong Kong, P. R. China} }

\bigskip

 \centerline{(Dedicated to Professor Qishao Lu with respect and admiration)}

\begin{abstract}
This paper concerns the consensus of discrete-time multi-agent
systems with linear or linearized dynamics. An observer-type
protocol based on the relative outputs of neighboring agents is
proposed. The consensus of such a multi-agent system with a directed
communication topology can be cast into the stability of a set of
matrices with the same low dimension as that of a single agent. The
notion of discrete-time consensus region is then introduced and
analyzed. For neurally stable agents, it is shown that there exists
an observer-type protocol having a bounded consensus region in the
form of an open unit disk, provided that each agent is stabilizable
and detectable. An algorithm is further presented to construct a
protocol to achieve consensus with respect to all the communication
topologies containing a spanning tree. Moreover, for unstable agents,
an algorithm is proposed to construct a protocol having an
origin-centered disk of radius $\delta$ ($0<\delta<1$) as its
consensus region,
where $\delta$ has to further satisfy a constraint
related to the unstable eigenvalues of a single agent for the case
where each agent has
a least one eigenvalue outside the unit circle.
Finally, the consensus algorithms are applied to
solve formation control problems of multi-agent systems.
\end{abstract}

\section{Introduction}

In recent years, the consensus issue of multi-agent systems has
received compelling attention from various scientific communities,
for its broad applications in such broad areas as satellite
formation flying, cooperative unmanned air vehicles, and air traffic
control, to name just a few. In \cite{vic95}, a simple model is
proposed for phase transition of a group of self-driven particles
with numerical demonstration of the complexity of the model. In
\cite{jad03}, it provides a theoretical explanation for the behavior
observed in \cite{vic95} by using graph theory. In \cite{sab04}, a
general framework of the consensus problem for networks of dynamic
agents with fixed or switching topologies is addressed. The
conditions given by \cite{sab04} are further relaxed in
\cite{ren052}. In \cite{hong} and \cite{hong2}, tracking control for
multi-agent consensus with an active leader is considered, where a
local controller is designed together with a neighbor-based
state-estimation rule. Some predictive mechanisms are introduced in
\cite{zhang} to achieve ultrafast consensus. In \cite{li2009,
linjia}, the $H_\infty$ consensus and control problems for networks
of agents with external disturbances and model uncertainties are
investigated. The consensus problems of networks of
double-integrator or high-order integrator agents are studied in
\cite{lin2, ren072, ren08, sun, tian, xie}. A distributed algorithm
is proposed in \cite{cortes} to asymptotically achieve consensus in
finite time. The so-called $\epsilon$-consensus problem is
considered in \cite{bau} for networks of dynamic agents with unknown
but bounded disturbances. The average agreement problem is examined
in \cite{fra} for a network of integrators with quantized links. The
controlled agreement problem of multi-agent networks is investigated
from a graph-theoretic perspective in \cite{erg}. Flocking
algorithms are investigated in \cite{sab06, su, tan} for a group of
autonomous agents. Another topic that is closely related to the
consensus of multi-agent systems is the synchronization of coupled
nonlinear oscillators, which has been extensively studied, e.g., in
\cite{bow, duan08, duan09, liu07, pe98, yam}. For a relatively
complete coverage of the literatures on consensus, readers are
referred to the recent surveys \cite{sab07, ren07}. In most existing
studies on consensus, the agent dynamics are restricted to be
first-, second-, and sometimes high-order integrators, and the
proposed consensus protocols are based on the relative states
between neighboring agents.

This paper considers the consensus of discrete-time linear
multi-agent systems with directed communication topologies. Previous
studies along this line include \cite{li2, li3, ma2, sup, seo, tun, tun2,
cheng}. In \cite{ma2, tun, tun2, cheng}, static consensus protocols
based on relative states of neighboring agents are used. The
discrete-time protocol in \cite{sup} requires the absolute output
measurement of each agent to be available, which is impractical in
many cases, e.g., the deep-space formation flying \cite{smh}.
Contrary to the protocol in \cite{sup}, an observer-type consensus
protocol is proposed here, based only on relative output
measurements of neighboring agents, which contains the static
consensus protocol developed in \cite{tun} as a special case. The
observer-type protocol proposed here can be seen as an extension of
the traditional observer-based controller for a single system to one
for the multi-agent systems. The Separation Principle of the
traditional observer-based controllers still holds in the
multi-agent setting presented in this paper.

More precisely, a decomposition approach is utilized here to convert
the consensus of a multi-agent system, whose communication topology
has a spanning tree, into the stability of a set of matrices with
the same dimension as a single agent. The final consensus value
reached by the agents is derived. Inspired by
the notion of continuous-time consensus region introduced in
\cite{li2} and the synchronized regions of complex networks studied
in \cite{duan08, liu07, pe98}, the notion of discrete-time consensus
region is introduced and analyzed. It is pointed out through
numerical examples that the consensus protocol should have a
reasonably large bounded consensus region so as to be robust to
variations of the communication topology. For the special case where the
state matrix is neutrally stable, it is shown that there exists an
observer-type protocol with a bounded consensus region in the form
of an open unit disk, if each agent is stabilizable and detectable.
An algorithm is further presented to construct a protocol to achieve
consensus with respect to all the communication topologies
containing a spanning tree. The main result in \cite{tun} can be
thereby easily obtained as a corollary. On the contrary, for the
general case where the state matrix is unstable, 
an algorithm is proposed to construct a protocol with the
origin-centered disk of radius $\delta$ ($0<\delta<1$) as its
consensus region. It is pointed out that $\delta$ has to
further satisfy a constraint
relying on the unstable eigenvalues of the state matrix
for the case where each agent has
a least one eigenvalue outside the unit circle,
which shows that the consensus problem of the discrete-time multi-agent systems
is generally more difficult to solve, compared to the continuous-time case in \cite{li2,li3}.

In the final, the consensus algorithms are modified to solve
formation control problems of multi-agent systems. Previous related
works include \cite{fax, laf, ren082}. In \cite{fax}, a Nyquist-type
criterion is presented to analyze the formation stability. The agent
dynamics in \cite{laf, ren082} are second-order integrators. In this
paper, a sufficient condition is given for the existence of a
distributed protocol to achieve a specified formation structure for
the multi-agent network, which generalizes the results in \cite{laf,
ren082}. Such a protocol can be constructed via the algorithms
proposed
as above. 

The rest of this paper is organized as follows. Notations and some
useful results of the graph theory is reviewed in Section 2. The
notion of discrete-time consensus region is introduced and analyzed
in Section 3. The special case where the state matrix is neutrally
stable is considered in Section 4. The case where the state
matrix is unstable is investigated in
Section 5. The consensus algorithms are applied to formation control
of multi-agent systems in Section 6. Section 7 concludes the paper.

\section{Notations and Preliminaries}

Let $\mathbf{R}^{n\times n}$ and $\mathbf{C}^{n\times n}$ be the
sets of $n\times n$ real matrices and complex matrices,
respectively. Matrices, if not explicitly stated, have compatible
dimensions in all settings. The superscript $T$ means transpose for
real matrices and $H$ means conjugate transpose for complex
matrices. $\|\cdot\|$ denotes the induced 2-norm. $I_N$ represents
the identity matrix of dimension $N$, and $I$ the identity matrix of
an appropriate dimension. Let $\mathbf{1}\in \mathbf{R}^p$ denote
the vector with all entries equal to one. For $\zeta\in\mathbf{C}$,
$\mathrm{Re}(\zeta)$ denotes its real part. $A\otimes B$ denotes the
Kronecker product of matrices $A$ and
$B$. 
The matrix inequality $A>B$ means that $A$ and $B$ are square
Hermitian matrices and $A-B$
is positive definite. 
A matrix $A\in\mathbf{C}^{n\times n}$ is neutrally stable in the
discrete-time sense if it has no eigenvalue with magnitude larger
than 1 and the Jordan block corresponding to any eigenvalue with
unit magnitude is of size one, while is Schur stable if all of its
eigenvalues have magnitude less than 1. A matrix
$Q\in\mathbf{R}^{n\times n}$ is orthogonal if $QQ^T=Q^TQ=I$. Matrix
$P\in\mathbf{R}^{n\times n}$ is an orthogonal projection onto the
subspace ${\rm{range}}(P)$ if $P^2 = P$ and $P^T = P$. Moreover,
${\rm{range}}(A)$ denotes the column space of matrix $A$, i.e, the
span of its column vectors.

A directed graph $\mathcal {G}$ is a pair $(\mathcal {V}, \mathcal
{E})$, where $\mathcal {V}$ is a nonempty finite set of nodes and
$\mathcal {E}\subset\mathcal {V}\times\mathcal {V}$ is a set of
edges, in which an edge is represented by an ordered pair of
distinct nodes. For an edge $(i,j)$, node $i$ is called the parent
node, $j$ the child node, and $j$ is neighboring to $i$. A graph
with the property that $(i,j)\in\mathcal {E}$ implies $(j,
i)\in\mathcal {E}$ is said to be undirected; otherwise, directed. A
path on $\mathcal {G}$ from node $i_1$ to node $i_l$ is a sequence
of ordered edges of the form $(i_k, i_{k+1})$, $k=1,\cdots,l-1$. A
directed graph has or contains a directed spanning tree if there
exists a node called root such that there exists a directed path
from this node to every other node in the graph.

For a graph $\mathcal {G}$ with $m$ nodes, the row-stochastic matrix
$\mathcal {D}\in\mathbf{R}^{m\times m}$ is defined with $d_{ii}>0$,
$d_{ij}>0$ if $(j,i)\in\mathcal {E}$ but $0$ otherwise, and
$\sum_{j=1}^md_{ij}=1$. According to \cite{ren052}, all of the
eigenvalues of $\mathcal {D}$ are either in the open unit disk or
equal to $1$, and furthermore, $1$ is a simple eigenvalue of
$\mathcal {D}$ if and only if graph $\mathcal {G}$ contains a
directed spanning tree. For an undirected graph, $\mathcal {D}$ is
symmetric.

Let $\Gamma_m$ denote the set of all directed graphs with $m$ nodes
such that each graph contains a directed spanning tree, and let
$\Gamma_{\leq \delta}$ ($0<\delta<1$) denote the set of all directed
graphs containing a directed spanning tree, whose non-one
eigenvalues lie in the disk of radius $\delta$ centered at the
origin.

\subsection{Problem Formulation}

Consider a network of $N$ identical agent with linear or linearized
dynamics in the discrete-time setting, where the dynamics of the
$i$-th agent are described by
\begin{equation}\label{1}
\begin{aligned}
    x_i^+ &=Ax_i+Bu_i,\\
    y_i & = Cx_i, \quad i=1,2,\cdots,N,
\end{aligned}
\end{equation}
where $x_i=x_i(k)=[x_{i,1},\cdots,x_{i,n}]\in\mathbf{R}^{n\times n}$
is the state, $x_i^+=x_i(k+1)$ is the state at the next time
instant, $u_i\in\mathbf{R}^{p}$ is the control input,
$y_i\in\mathbf{R}^{q}$ is the measured output, and $A$, $B$, $C$ are
constant matrices with compatible dimensions.

The communication topology among agents is represented by a directed
graph $\mathcal {G}=(\mathcal {V},\mathcal {E})$, where $\mathcal
{V}=\{1,\cdots,N\}$ is the set of nodes (i.e., agents) and $\mathcal
{E}\subset\mathcal {V}\times\mathcal {V}$ is the set of edges. An
edge $(i,j)$ in graph $\mathcal {G}$ means that agent $j$ can obtain
information from agent $i$, but not conversely.

At each time instant, the information available to agent $i$ is the
relative measurements of other agents with respect to itself, given
by
\begin{equation}\label{er} \zeta_i=\sum_{j=1}^Nd_{ij}(y_i-y_j),
\end{equation}
where $\mathcal {D}=(d_{ii})_{N\times N}$ is the row-stochastic
matrix associated with graph $\mathcal {G}$. A distributed
observer-type consensus protocol is proposed as
\begin{equation}\label{cl}
\begin{aligned}
v_i^+ &=(A+BK)v_i+L\left(\sum_{j=1}^Nd_{ij}C(v_i-v_j)-\zeta_i\right),\\
u_i &=Kv_i,
\end{aligned}
\end{equation}
where $v_i\in\mathbf{R}^{n}$ is the protocol state, $i=1,\cdots,N$,
$L\in\mathbf{R}^{q\times n}$ and $K\in\mathbf{R}^{p\times n}$ are
feedback gain matrices to be determined. In \dref{cl}, the term
$\sum_{j=1}^Nd_{ij}C(v_i-v_j)$ denotes the information exchanges
between the protocol of agent $i$ and those of its neighboring
agents. It is observed that the protocol \dref{cl} maintains the
same communication topology as the agents in \dref{1}.

Let $z_i=[x_i^T,v_i^T]^T$ and $z=[z_1^T,\cdots,z_N^T]^T$. Then, the
closed-loop system resulting from \dref{1} and \dref{cl} can be
written as
\begin{equation}\label{netg}
\begin{aligned}
z^+=(I_N\otimes\mathcal {A}+(I_N-\mathcal {D})\otimes\mathcal {H})z,
\end{aligned}
\end{equation}
where
$$\mathcal {A}=\begin{bmatrix}A & BK\\0 &
A+BK\end{bmatrix}, \quad \mathcal {H}=\begin{bmatrix}0 & 0\\
-LC & LC\end{bmatrix}.$$

\begin{definition} Given agents \dref{1}, the protocol
\dref{cl} is said to solve the consensus problem if
\begin{equation}\label{con}
\|x_i(k)- x_j(k)\|\rightarrow 0,~\text{as} ~k\rightarrow \infty,
~\forall~i,j=1,2,\cdots,N.
\end{equation}
\end{definition}

The following presents a decomposition approach to the consensus
problem of network \dref{netg}.

\begin{theorem} \label{th1}
        For any $\mathcal {G}\in\Gamma_N$, the
agents in \dref{1} reach consensus under protocol \dref{cl} if all
the matrices $A+BK$, $A+(1-\lambda_i)LC$, $i=2,\cdots,N$, are Schur
stable, where $\lambda_i$, $i=2,\cdots,N$, denote the eigenvalues of
$\mathcal {D}$ located in the open unit disk.
\end{theorem}

\begin{proof}
For any $\mathcal {G}\in\Gamma_N$, it is known that $0$ is a simple
eigenvalue of $I_N-\mathcal {D}$ and the other eigenvalues lie in
the open unit disk centered at $1+{\rm i}0$ in the complex plane,
where ${\rm i}=\sqrt{-1}$. Let $r^T\in\mathbf{R}^{1\times N}$ be the
left eigenvector of $I_N-\mathcal {D}$ associated with the
eigenvalue $0$, satisfying $r^T{\bf 1}=1$. Introduce
$\xi\in\mathbf{R}^{2Nn\times 2Nn}$ by
\begin{equation}\label{au}
\begin{aligned}
\xi(t) &=z(t)-\left(({\bf 1}r^T)\otimes I_{2n}\right)z(t)\\
 &= \left((I_N-{\bf 1}r^T)\otimes I_{2n}\right)z(t),
\end{aligned}
\end{equation}
which satisfies $(r^T\otimes I_{2n})\xi=0$. It is easy to see that
$0$ is a simple eigenvalue of $I_N-{\bf 1}r^T$ with $\mathbf{1}$ as
its right eigenvector, and 1 is another eigenvalue with multiplicity
$N-1$. Thus, it follows from \dref{au} that $\xi=0$ if and only if
$z_1=z_2=\cdots=z_N$, i.e., the consensus problem can be cast into
the Schur stability of vector $\xi$, which evolves according to the
following dynamics:
\begin{equation}\label{au1}
\begin{aligned}
\xi^+=(I_N\otimes\mathcal {A}+(I-\mathcal {D})\otimes\mathcal
{H})\xi.
\end{aligned}
\end{equation}

Next, let $Y\in\mathbf{R}^{N\times(N-1)}$,
$W\in\mathbf{R}^{(N-1)\times N}$, $T\in\mathbf{R}^{N\times N}$, and
upper-triangular $\Delta\in\mathbf{R}^{(N-1)\times(N-1)}$ be such
that
\begin{equation}\label{djo}
T=\begin{bmatrix} \mathbf{1} & Y
\end{bmatrix},  \quad T^{-1}=\begin{bmatrix} r^T \\ W
\end{bmatrix}, \quad T^{-1}(I_N-\mathcal
{D})T=J=\begin{bmatrix} 0 & 0 \\ 0 & \Delta\end{bmatrix},
\end{equation}
where the diagonal entries of $\Delta$ are the nonzero eigenvalues
of $I_N-\mathcal {D}$. Introduce the state transformation
$\zeta=(T^{-1}\otimes I_{2n})\xi$ with
$\zeta=[\zeta_1^T,\cdots,\zeta_N^T]^T$. Then, \dref{au1} can be
represented in terms of $\zeta$ as follows:
\begin{equation}\label{dnetet}
\begin{aligned}
\zeta^+=(I_N\otimes\mathcal {A}+J\otimes\mathcal {H})\zeta.
\end{aligned}
\end{equation}
As to $\zeta_1$, it can be seen from \dref{au} that
\begin{equation}\label{dt10}
\zeta_1=(r^T\otimes I_{2n})\xi\equiv0.
\end{equation}
Note that the elements of the state matrix of \dref{dnetet} are
either block diagonal or block upper-triangular. Hence, $\zeta_i,$
$i=2,\cdots,N$, converge asymptotically to zero if and only if the
$N-1$ subsystems along the diagonal, i.e.,
\begin{equation}\label{dt11}
\zeta^+_i =(\mathcal {A}+(1-\lambda_i)\mathcal {H})\zeta_i,\quad
i=2,\cdots,N,
\end{equation}
are Schur stable. It is easy to verify that matrices $\mathcal
{A}+\lambda_i\mathcal {H}$ are similar to
$$\begin{bmatrix} A+(1-\lambda_i) LC &0 \\
-(1-\lambda_i) LC & A + BK \end{bmatrix},\quad i=2,\cdots,N. $$
Therefore, the Schur stability of the matrices $A+BK$,
$A+(1-\lambda_i)LC$, $i=2,\cdots,N$, is equivalent to that the state
$\zeta$ of \dref{au1} converges asymptotically to zero, implying
that consensus is achieved.
\end{proof}

\begin{remark}
The importance of this theorem lies in that it converts the
consensus problem of a large-scale therefore very high-dimensional
multi-agent network under the observer-type protocol \dref{cl} to
the stability of a set of matrices with the same dimension as a
single agent, thereby significantly reducing the computational
complexity. The directed communication topology $\mathcal {G}$ is
only assumed to have a directed
spanning tree.
The effects of the communication topology on the consensus problem
are characterized by the eigenvalues of the corresponding
row-stochastic matrix $\mathcal {D}$, which may be complex,
rendering the matrices be complex-valued in Theorem \ref{th1}.
\end{remark}

\begin{remark}
The observer-type consensus protocol \dref{cl} can be seen as an
extension of the traditional observer-based controller for a single
system to one for multi-agent systems. The Separation Principle of
the traditional observer-based controllers still holds in this
multi-agent setting. Moreover, the protocol \dref{cl} is based only
on relative output measurements between neighboring agents, which
can be regarded as the discrete-time counterpart of the protocol
proposed in \cite{li2,li3}, including the static protocol used in
\cite{tun} as a special case.
\end{remark}

\begin{theorem} \label{th2}
Consider the multi-agent network
\dref{netg} with a communication topology $\mathcal {G}\in\Gamma_N$.
If protocol \dref{cl} satisfies Theorem \ref{th1}, then
\begin{equation}\label{cv}
\begin{aligned}
x_i(k)&\rightarrow \varpi(k)\triangleq(r^T\otimes
A^k)\begin{bmatrix} x_1(0)\\\vdots
\\ x_N(0)\end{bmatrix}, \\
v_i(k)&\rightarrow 0,~i=1,2,\cdots,N,~ \text{ as} ~k\rightarrow
\infty,
\end{aligned}
\end{equation}
where $r\in \mathbf{R}^{N}$ satisfies $r^T(I_N-\mathcal {D})=0$ and
$r^T{\bf 1}=1$.
\end{theorem}

\begin{proof}
The solution of \dref{netg} can be obtained as
$$\begin{aligned} z(k+1)  &=\left(I_N\otimes\mathcal
{A}+(I_N-\mathcal {D})
          \otimes\mathcal{H}\right)^kz(0) \\
        &=(T\otimes I)(I_N\otimes\mathcal
          {A}+J\otimes\mathcal{H})^k(T^{-1}\otimes I)z(0)\\
        &=(T\otimes I)\begin{bmatrix}\mathcal
            {A}^k & 0 \\ 0 & (I_{ N-1}\otimes\mathcal
            {A}+\Delta\otimes\mathcal {H})^k\end{bmatrix}(T^{-1}\otimes I)z(0),
\end{aligned}$$
where matrices $T$, $J$ and $\Delta$ are defined in \dref{djo}. By
Theorem \ref{th1}, $I_{N-1}\otimes\mathcal {A}+\Delta\otimes\mathcal
{H}$ is Schur stable. Thus,
$$\begin{aligned}
z(k+1) &\rightarrow(\mathbf{1}\otimes I)\mathcal
{A}^k(r^T\otimes I)z(0) \\
&=(\mathbf{1}r^T)\otimes \mathcal {A}^kz(0),~\text{ as}
~k\rightarrow \infty,
\end{aligned}$$ implying that
\begin{equation}\label{scne2}
z_i(k)\rightarrow (r^T\otimes \mathcal {A}^k)z(0), ~\text{ as}
~k\rightarrow \infty,~i=1,\cdots,N.
\end{equation} Since $A+BK$ is Schur stable,
\dref{scne2} directly leads to the assertion.
\end{proof}

\begin{remark}
Some observations on the final consensus
value in \dref{cv} can be concluded as follows:
If $A$ is Schur stable, then
$\varpi(k)\rightarrow 0$, as $k\rightarrow \infty$. If $A$ in \dref{1}
has eigenvalues located outside the open unit circle, then the
consensus value $\varpi(k)$ reached by the agents will tend to
infinity exponentially. On the other
hand, if $A$ has eigenvalues in the closed unit circle, then the
agents in \dref{1} may reach consensus nontrivially. That is, some
states of each agent might approach a common nonzero value. Typical
examples belonging to the last case include the commonly-studied
first-, second-, and high-order integrators.
\end{remark}

\section{Discrete-Time Consensus Regions}

From Theorem \dref{th1}, it can be noticed that the consensus of the
given agents \dref{1} under protocol \dref{cl} depends on the
feedback gain matrices $K$, $L$, and the eigenvalues $\lambda_i$ of
matrix $\mathcal {D}$ associated with the communication graph
$\mathcal {G}$, where matrix $L$ is coupled with $\lambda_i$,
$i=2,\cdots,N$. Hence, it is useful to analyze the correlated
effects of matrix $L$ and graph $\mathcal {G}$ on consensus. To this
end, the notion of consensus region is introduced.

\noindent{\small \bf Definition 2}. Assume that matrix $K$ has been designed
such that $A+BK$ is Schur stable. The region $\mathcal {S}$ of the
parameter $\sigma\subset\mathbf{C}$, such that matrix $A+(1-\sigma)
LC$ is Schur stable, is called the (discrete-time) consensus region
of network \dref{netg}.

The notion of discrete-time consensus region is inspired by the
continuous-time consensus region introduced in \cite{li2} and the
synchronized regions of complex networks studied in \cite{duan08,
liu07, pe98}. The following result is a direct consequence of
Theorem \ref{th1}.

\begin{corollary} \label{cor1} The agents in \dref{1} reach consensus
under protocol \dref{cl} if $\lambda_i\in\mathcal {S}$,
$i=2,\cdots,N$, where $\lambda_i$, $i=2,\cdots,N$, are the
eigenvalues of $\mathcal {D}$ located in the open unit disk.
\end{corollary}

For an undirected communication graph, the consensus region of
network \dref{netg} is a bounded interval or a union of several
intervals on the real axis. However, for a directed graph where the
eigenvalues of $\mathcal {D}$ are generally complex numbers, the
consensus region $\mathcal {S}$ is either a bounded region or a set
of several disconnected regions in the complex plane. Due to the
fact that the eigenvalues of the row-stochastic matrix $\mathcal
{D}$ lie in the unit disk, unbounded consensus regions, desirable
for consensus in the continuous-time setting \cite{li2,li3}, generally
do not exist for the discrete-time consensus considered here.

The following example has a disconnected consensus region.

\noindent{\small \bf Example 1}. The agent dynamics and the consensus
protocol are given by \dref{1} and \dref{cl}, respectively, with
$$\begin{aligned}
A&=\begin{bmatrix}0 & 1 \\ -1 & 1.02\end{bmatrix},\quad
B=\begin{bmatrix} 1 \\ 0\end{bmatrix},\quad C=\begin{bmatrix} 1 & 0
\\ 0 & 1\end{bmatrix},\\
 L&=\begin{bmatrix} 0 & -1 \\ 1 &
0\end{bmatrix},  \quad K=\begin{bmatrix} -0.5 & -0.5\end{bmatrix}.
\end{aligned}$$
Clearly, matrix $A+BK$ with $K$ given as above is Schur stable. For
simplicity in illustration, assume that the communication graph
$\mathcal {G}$ is undirected here. Then, the consensus region is a
set of intervals on the real axis. The characteristic equation of
$A+(1-\sigma) LC$ is
\begin{equation}\label{dcp}
{\rm{det}}(zI-A-\sigma LC)=z^2-1.02z+\sigma^2=0.
\end{equation}
Applying bilinear transformation $z = \frac{s+1}{s-1}$ to \dref{dcp}
gives
\begin{equation}\label{dcp2}
(\sigma^2-0.02)s^2+(1-\sigma^2)s+2.02+\sigma^2=0.
\end{equation}
It is well known that, under the bilinear transformation, \dref{dcp}
has all roots within the unit disk if and only if the roots of
\dref{dcp2} lie in the open left-half plane (LHP). According to the
Hurwitz criterion \cite{ogata}, \dref{dcp2} has all roots in the
open LHP if and only if $0.02<\sigma^2<1$. Therefore, the consensus
region in this case is $\mathcal {S}=(-1,-0.1414)\cup (0.1414, 1)$,
a union of two disconnected intervals. For the communication graph
shown in Figure 1, the corresponding row-stochastic matrix is
$$\mathcal {D}=\begin{bmatrix}0.3& 0.2 & 0.2 & 0.2 & 0 & 0.1
\\ 0.2 & 0.6 & 0.2 &0 &0 &0\\
0.2 & 0.2 & 0.6 & 0 & 0& 0\\ 0.2 & 0 & 0 & 0.4 & 0.4 & 0\\ 0 & 0 & 0
& 0.4 &0.2 &0.4\\ 0.1 & 0 & 0 & 0 & 0.4 & 0.5\end{bmatrix},
$$ whose eigenvalues, other than 1, are
$-0.2935,0.164,0.4,0.4624,0.868$, which all belong to $\mathcal
{S}$. Thus, it follows from Corollary \ref{cor1} that network
\dref{netg} with graph given in Figure 1 can achieve consensus.

\begin{figure}[htbp] \centering
\includegraphics[width=2.2in]{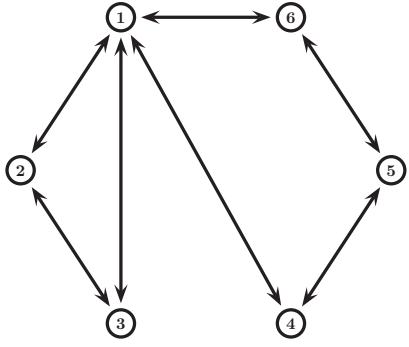}
\caption{The communication topology.}
\end{figure}

Let's see how modifications of the communication topology affect the
consensus. Consider the following two simple cases:

1) An edge is added between nodes 1 and 5, thus
    more information exchange will exist inside the network. Then, the
    row-stochastic matrix $\mathcal {D}$ becomes
    $$\begin{bmatrix}0.2& 0.2 & 0.2 & 0.2 & 0.1 & 0.1
\\ 0.2 & 0.6 & 0.2 &0 &0 &0\\
0.2 & 0.2 & 0.6 & 0 & 0& 0\\ 0.2 & 0 & 0 & 0.4 & 0.4 & 0\\ 0.1 & 0 &
0 & 0.4 &0.2 &0.3\\ 0.1 & 0 & 0 & 0 & 0.4 & 0.5\end{bmatrix},$$
whose eigenvalues, in addition to 1, are
$-0.2346,0.0352,0.4,0.4634,0.836$. Clearly, the eigenvalue $0.0352$
does not belong to $\mathcal {S}$, i.e., consensus can not be
achieved in this case.

 2) The edge between nodes 5 and 6 is removed. The row-stochastic
 matrix $\mathcal {D}$
 becomes
 $$\begin{bmatrix}0.3& 0.2 & 0.2 & 0.2 & 0 & 0.1
\\ 0.2 & 0.6 & 0.2 &0 &0 &0\\
0.2 & 0.2 & 0.6 & 0 & 0& 0\\ 0.2 & 0 & 0 & 0.4 & 0.4 & 0\\ 0.1 & 0 &
0 & 0.4 &0.6 &0\\ 0.1 & 0 & 0 & 0 & 0 & 0.9\end{bmatrix},$$ whose
eigenvalues, other than 1, are $-0.0315, 0.2587,    0.4,    0.8676,
0.9052$. In this case, the eigenvalue $-0.0315$ does not belong to
$\mathcal {S}$, i.e., consensus can not be achieved either.

These sample cases imply that, for disconnected consensus regions,
consensus can be quite fragile to the variations of the network's
communication topology. Hence, the consensus protocol should be
designed to have a sufficiently large bounded consensus region in
order to be robust with respect to the communication topology. This
is the topic of the following sections.

\section{Networks with Neurally Stable Agents}

In this section, a special case where matrix $A$ is neutrally stable
is considered. First, the following lemma is needed.

\begin{lemma} [\cite{zhou}] \label{lem1} For matrix $Q=Q^H\in
\mathbf{C}^{n\times n}$, consider the following Lyapunov equation:
$$A^HXA-X+Q=0.$$
If $X>0$, $Q\geq 0$, and $(Q,V)$ is observable, then matrix $A$ is
Schur stable .
\end{lemma}

\begin{proposition} \label{pro1}
For matrices $Q\in \mathbf{R}^{n\times n}$, $V\in
\mathbf{R}^{m\times n}$, $\sigma\in\mathbf{C}$, where $Q$ is
orthogonal, $VV^T=I$, and $(Q,V)$ is observable, if $|\sigma|< 1$,
then the matrix $Q-(1-\sigma)QV^TV$ is Schur stable.
\end{proposition}

\begin{proof}
Observe that
\begin{equation}\label{l41}
\begin{aligned}
&(Q-(1-\sigma)QV^TV)^{H}(Q-(1-\sigma)QV^TV)-I \\
         &\qquad = Q^TQ-(1-\sigma)Q^TQV^TV-(1-\bar{\sigma})V^TVQ^TQ \\
         &\qquad\quad +|1-\sigma|^2V^TVQ^TQV^TV-I\\
         &\qquad= (-2\mathbf{Re}(1-\sigma)+|1-\sigma|^2)V^TV\\
         &\qquad= (|\sigma|^2-1)V^TV.
\end{aligned}
\end{equation}
Since $(Q,V)$ is observable, it is easy to verify that
$(Q-(1-\sigma)QV^TV,V^TV)$ is also observable. Then, by Lemma
\ref{lem1}, \dref{l41} implies that $Q-(1-\sigma)QV^TV$ is Schur
stable for any $|\sigma|< 1$.
\end{proof}

Next, an algorithm for protocol \dref{cl} is presented, which will
be used later.

\noindent{\bf Algorithm 1}. Given that $A$ is neutrally stable and
that $(A,B,C)$ is stabilizable and detectable, the protocol
\dref{cl} can be constructed as follows:
\begin{itemize}
\item[1)] Select $K$ be such that $A+BK$ is Schur stable.

\item[2)] Choose $U\in\mathbf{R}^{n\times n_1}$ and
$W\in\mathbf{R}^{n\times(n-n_1)}$, satisfying \footnote{Matrices $U$
and $W$ can be derived by transforming matrix $A$ into the real
Jordan canonical form \cite{horn85}.}
\begin{equation}\label{a11}
\begin{bmatrix} U & W \end{bmatrix}^{-1}A
\begin{bmatrix} U & W \end{bmatrix}=
\begin{bmatrix} M & 0 \\ 0 & X \end{bmatrix},
\end{equation}
where $M\in\mathbf{R}^{n_1\times n_1}$ is orthogonal and
$X\in\mathbf{R}^{(n-n_1)\times (n-n_1)}$ is Schur stable.

\item[3)] Choose $V\in\mathbf{R}^{m\times n_1} $ such that $VV^T=I_m$ and
$\mathrm{range}(V^T)=\mathrm{range}(U^TC^T)$. 

\item[4)] Define $L=-UMV^T(CUV^T)^{-1}$.
\end{itemize}

\begin{theorem}\label{th3}
Suppose that matrix $A$ is neutrally stable and that $(A,B,C)$ is
stabilizable and detectable. The protocol \dref{cl} constructed via
Algorithm 1 has the open unit disk as its bounded consensus region.
Thus, such a protocol solves the consensus problem for \dref{1} with
respect to $\Gamma_N$, the set of all the communication topologies
containing a spanning tree.
\end{theorem}

\begin{proof}
Let the related variables be defined as in Algorithm 1. Assume
without loss of generality that matrix $CU$ is of full row rank.
Since $V^TV$ is an orthogonal projection onto
$\mathrm{range}(V^T)=\mathrm{range}(U^TC^T)$, matrix $CUV^T$ is
invertible and $V^TVU^TC^T=U^TC^T$, so that $V= (CU V^T)^{-1}CU$,
and hence $LCU =-UMV^TV$. Also, the detectability of $( A,C)$
implies that $(M,V )$ is observable. Let $U^{\dag}\in
\mathbf{R}^{n_1\times n}$ and $W^{\dag}\in \mathbf{R}^{(n-n_1)\times
n}$ be such that $\begin{bmatrix} U^{\dag}
\\ W^{\dag}
\end{bmatrix}=\begin{bmatrix} U & W \end{bmatrix}^{-1},$
where $U^{\dag}U=I$, $W^{\dag}W=I$, $U^{\dag}W=0$, and
$W^{\dag}U=0$. Then,
\begin{equation}\label{dl31}
\begin{aligned}
&\begin{bmatrix} U & U \end{bmatrix}^{-1}(A+(1-\sigma)LC)
\begin{bmatrix} U & W \end{bmatrix} \\
        &= \begin{bmatrix} M+(1-\sigma)U^{\dag}LCU & (1-\sigma)U^{\dag}LCW \\
        (1-\sigma)W^{\dag}LCU&
        X+(1-\sigma)W^{\dag}LCW\end{bmatrix}\\
        &= \begin{bmatrix} M-(1-\sigma)MV^TV & -(1-\sigma)U^{\dag}LCW \\ 0 &
        X\end{bmatrix}.
\end{aligned}
\end{equation}
By Lemma \ref{lem1}, matrix $M-(1-\sigma)MV^TV$ is Schur stable for
any $|\sigma|< 1$. Hence, \dref{dl31} implies that matrix
$A+(1-\sigma)LC$ with $L$ given by Algorithm 1 is Schur stable for
any $|\sigma|< 1$, i.e., the protocol \dref{cl} constructed via
Algorithm 1 has a bounded consensus region in the form of the open
unit disk. Since the  eigenvalues of any communication topology
containing a spanning tree lie in the open unit disk, except
eigenvalue 1, it follows from Corollary \ref{cor1} that this
protocol solves the consensus problem with respect to $\Gamma_N$.
\end{proof}

In \cite{tun}, the consensus of the following coupled network is
considered:
\begin{equation}\label{tu}
x_i^+ =Ax_i+LC\sum_{j=1}^Nd_{ij}(x_i-x_j), \quad i=1,2\cdots,N,
\end{equation}
where $(d_{ij})_{N\times N}$ is defined as in \dref{er} and matrix
$L$ is to be designed.

The main result of \cite{tun} can be easily obtained as a corollary
here.

\begin{corollary} \label{cor2} There exists a matrix $L$ such that
network \dref{tu} has the open unit disk as its consensus region,
i.e., the network can reach consensus with respect to $\Gamma_N$, if
and only if the pair $(A,C)$ is detectable. Such a matrix $L$ can be
constructed via Algorithm 1.
\end{corollary}

\begin{proof} For any communication topology $\mathcal {G}\in\Gamma_N$, it follows from Theorem
\ref{th1} that there exists a matrix $L$ such that network \dref{tu}
achieves consensus if and only if matrices $A+(1-\lambda_i)LC$,
$i=2,\cdots,N$, are Schur stable. The necessity is obvious and the
sufficiency follows readily from Theorem \ref{th3}.
\end{proof}

\begin{remark} Compared to Theorem 6 in \cite{tun}, the
above corollary presents a necessary and sufficient condition for
the existence of a matrix $L$ that ensures the network to reach
consensus. Moreover, the method leading to this corollary is quite
different from and comparatively much simpler than that used in
\cite{tun}. Of course, it should be admitted that the proof of
Theorem \ref{th3} above is partly inspired by \cite{tun}.
\end{remark}

\noindent{\bf Example 2}. Consider a network of agents described by
\dref{1}, with
$$A=\begin{bmatrix} 0.2 & 0.6 & 0\\ -1.4& 0.8& 0\\0.7 & 0.2
&-0.5\end{bmatrix},\quad B=\begin{bmatrix} 0 \\ 1 \\ 0\end{bmatrix},
\quad C=\begin{bmatrix} 1 & 0 & 1\end{bmatrix}.$$
The eigenvalues of matrix $A$ are $-0.5$, $0.5\pm {\rm i}0.866$,
thus $A$ is neutrally stable. In protocol \dref{cl}, choose
$K=\left[\begin{matrix}1.2 &-0.9 &-0.2\end{matrix}\right]$ such that
$A+BK$ is Schur stable. The matrices
$$U=\begin{bmatrix} 0.1709 & -0.4935 \\ 0.7977  &  0\\-0.0570 & -0.2961\end{bmatrix},
\quad W=\begin{bmatrix} 0 \\ 0 \\ 1\end{bmatrix}$$ satisfy
\dref{a11} with $M=\left[\begin{smallmatrix} 0.5 & 0.866
\\-0.866 & 0.5\end{smallmatrix}\right]$ and $X=-0.5$. Thus,
$U^TC^T=\left[\begin{smallmatrix}0.1139& -0.7896\end{smallmatrix}\right]^T$.
Take $V=\left[\begin{matrix}0.1428& -0.9898\end{matrix}\right]$ such
that $VV^T=1$ and ${\rm{range}}(V^T)={\rm{range}}(U^TC^T)$. Then, by
Algorithm 1, one obtains $L=\left[\begin{matrix}-0.2143& 0.7857 &
-0.2857\end{matrix}\right]^T$. In light of Theorem \ref{th3}, the
agents considered in this example will reach consensus under the
protocol \dref{cl}, with $K$ and $L$ given as above, with respect to
all the communication topologies containing a spanning tree.

\section{Networks with Unstable Agents}

This section considers the general case where matrix $A$ is not neutrally stable,
i.e., $A$ is allowed to have eigenvalues outside the unit circle
or has at least one eigenvalue with unit magnitude whose corresponding Jordan
block is of size larger than 1.

Before moving forward, one introduces the following modified algebraic
Riccati equation (MARE) \cite{kat, sch, sino}:
\begin{equation}\label{ric}
P=APA^T-(1-\delta^2) APC^T(CPC^T+I)^{-1}CPA^T+Q,
\end{equation}
where $P\geq0$, $Q>0$, and $\delta\in\mathbf{R}$. For $\delta=0$, the
MARE \dref{ric} is reduced to the commonly-used discrete-time
Riccati equation discussed in, e.g., \cite{zhou}.

The following lemma concerns the existence of solutions for the
MARE.

\begin{lemma}[\cite{sch, sino}]\label{lric}
Let $(A,C)$ be detectable.
Then, the following statements hold.
\begin{itemize}
\item[a)]
Suppose that the matrix $A$ has no eigenvalues with magnitude larger than 1,
Then, the MARE \dref{ric} has a unique positive-definite solution $P$ for any $0<\delta<1$.

\item[b)]
For the case where $A$ has a least one
eigenvalue with magnitude larger than 1 and the rank of $B$ is one,
the MARE \dref{ric} has a unique positive-definite solution $P$,
if $0<\delta<\frac{1}{\underset{i}{\Pi}|\lambda_i^u(A)|}$, where
$\lambda_i^u(A)$ denote
the unstable eigenvalues of $A$.

\item[c)]
If the MARE \dref{ric} has a unique positive-definite solution $P$, then
$P=\lim_{k\rightarrow \infty}P_k$ for any initial condition
$P_0\geq 0$, where $P_k$ satisfies
$$P(k+1)=AP(k)A^T-(1-\delta^2) AP(k)C^T(CP(k)C^T+I)^{-1}CP(k)A^T+Q.$$
\end{itemize}
\end{lemma}

\begin{proposition} \label{pro3}
Suppose that $(A,C)$ be detectable. Then, for the case where $A$ has no eigenvalues with
magnitude larger than 1, the matrix $A+(1-\sigma) LC$
with $L=-APC^T(CPC^T+I)^{-1}$ is Schur stable for any
$|\sigma|\leq\delta$, $0<\delta<1$, where $P>0$ is the unique
solution to the MARE \dref{ric}. Moreover,
for the case where $A$ has at least eigenvalue with
magnitude larger than 1 and $B$ is of rank one,
$A+(1-\sigma) LC$
with $L=-APC^T(CPC^T+I)^{-1}$ is Schur stable for any
$|\sigma|\leq\delta$, $0<\delta<\frac{1}{\underset{i}{\Pi}|\lambda_i^u(A)|}$.
\end{proposition}

\begin{proof}
Observe that
\begin{equation}\label{ricp}
\begin{aligned}
&(A+(1-\sigma) LC)P(A+(1-\sigma) LC)^H-P\\
&=APA^T-2\mathbf{Re}(1-\sigma)APC^T(CPC^T+I)^{-1}CPA^T-P\\
&\quad+ |1-\sigma|^2APC^T(CPC^T+I)^{-1}CPC^T(CPC^T+I)^{-1}CPA^T\\
&=APA^T+(-2\mathbf{Re}(1-\sigma)+|1-\sigma|^2)APC^T(CPC^T+I)^{-1}CPA^T-P\\
&\quad +|1-\sigma|^2APC^T(CPC^T+I)^{-1}\left(-I+CPC^T(CPC^T+I)^{-1}\right)CPA^T\\
&=APA^T+(|\sigma|^2-1)APC^T(CPC^T+I)^{-1}CPA^T-P\\
&\quad -|1-\sigma|^2APC^T(CPC^T+I)^{-2}CPA^T\\
&\leq APA^T-(1-\delta^2) APC^T(CPC^T+I)^{-1}CPA^T-P\\
&=-Q<0,
\end{aligned}\end{equation}
where the identity $-I+CPC^T(CPC^T+I)^{-1}=-(CPC^T+I)^{-1}$ has been
applied. Then, the assertion follows directly from Lemma \ref{lric}
and the discrete-time Lyapunov inequality.
\end{proof}

\noindent{\bf Algorithm 2}. Assuming that $(A,B,C)$ is stabilizable
and detectable, the protocol \dref{cl}
can the constructed as follows:

\begin{itemize}
\item[1)] Select $K$ such that $A+BK$ is Schur stable.

\item[2)] Choose $L=-APC^T(CPC^T+I)^{-1}$, where $P>0$ is the unique
solution of \dref{ric}.
\end{itemize}

\begin{remark}
By Lemma \ref{lric} and Proposition \ref{pro3}, it follows that
a sufficient and necessary condition
for the existence of the consensus protocol
by using Algorithm 2 is that $(A,B,C)$ is
stabilizable and detectable
for the case where $A$ has no eigenvalues with
magnitude larger than 1. In contrast,
$\delta$ has to further satisfy $\delta<\frac{1}{\underset{i}{\Pi}|\lambda_i^u(A)|}$
for the case where $A$ has at least eigenvalue outside
the unit circle and $B$ is of rank one.
\end{remark}

\vspace*{8pt}

The result below follows directly from Theorem \ref{th2} and
Proposition \ref{pro3}.

\begin{theorem}\label{th4}
Let $(A,B,C)$ be stabilizable
and detectable. Then, the protocol given by Algorithm 2 has a bounded
consensus region in the form of an origin-centered disk of radius
$\delta$, i.e., this protocol solves the consensus problem
for networks with agents \dref{1} with respect to $\Gamma_{\leq
\delta}$, where $\delta$ satisfies $0<\delta<1$ for the case where
$A$ has no eigenvalues with magnitude larger than 1 and
satisfies $0<\delta<\frac{1}{\underset{i}{\Pi}|\lambda_i^u(A)|}$ for
the case where $A$ has a least one eigenvalue outside the unit circle
and $B$ is of rank one.
\end{theorem}

\begin{remark}
Note that $\Gamma_{\leq \delta}$ was defined in Section 2,
which is a subset of $\Gamma_m$ in the special case where $A$ is
neutrally stable as discussed in the above section. This is
consistent with the intuition that unstable behaviors are more
difficult to synchronize than the neutrally stable ones.
\end{remark}

\noindent{\bf Example 3}. Let the agents in \dref{1} be
discrete-time double integrators, with
$$A=\begin{bmatrix} 1& 1\\ 0 & 1 \end{bmatrix},\quad
B=\begin{bmatrix} 0 \\ 1 \end{bmatrix},\quad C=\begin{bmatrix} 1 & 0
\end{bmatrix}.$$ Obviously, Assumption 1 holds here. Choose
$K=\left[\begin{matrix} -0.5& -1.5\end{matrix}\right]$, so that
matrix $A+BK$ is Schur stable. Solving equation \dref{ric} with
$\delta=0.95$ gives
$P=10^4\times\begin{bmatrix} 1.1780&0.0602\\
0.0602 &0.0062
\end{bmatrix}$. By Algorithm 2, one obtains $L=\left[\begin{matrix} -1.051& -0.051\end{matrix}\right]^T$.
 It follows from Theorem \ref{th4} that the agents \dref{1} reach consensus under
protocol \dref{cl} with $K$ and $L$ given as above with
 respect to $\Gamma_{\leq 0.95}$. 
Assume that the communication topology $\mathcal {G}$ is given as in
Figure 2, and the corresponding row-stochastic matrix is
$$\mathcal {D}=\begin{bmatrix}0.4& 0 & 0 & 0.1 & 0.3& 0.2\\ 0.5 & 0.5 & 0 &0 &0 &0\\
0.3& 0.2 & 0.5 & 0 & 0& 0\\ 0.5 & 0 & 0 & 0.5 & 0 & 0\\ 0 & 0 & 0 &
0.4 & 0.4 & 0.2\\ 0 & 0 & 0 & 0 & 0.3 & 0.7\end{bmatrix},
$$ whose eigenvalues, other than 1, are
$\lambda_i=0.5,0.5565,0.2217 \pm {\rm i}0.2531$. Clearly,
$|\lambda_i|<0.95$, for $i=2,\cdots,6$. Figure 3 depicts the state
trajectories of network \dref{netg} for this example, which shows
that consensus is actually achieved.

\begin{figure}[htp]
\begin{center}
  \includegraphics[width=2.2in]{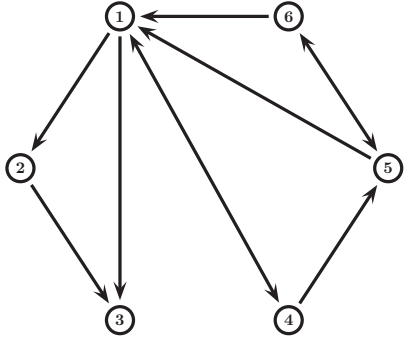}\\
  \caption{The communication topology.}\label{AIMS}
  \end{center}
\end{figure}

\begin{figure}[htp]
\begin{center}
\includegraphics[width=0.8\linewidth,height=0.3\linewidth]{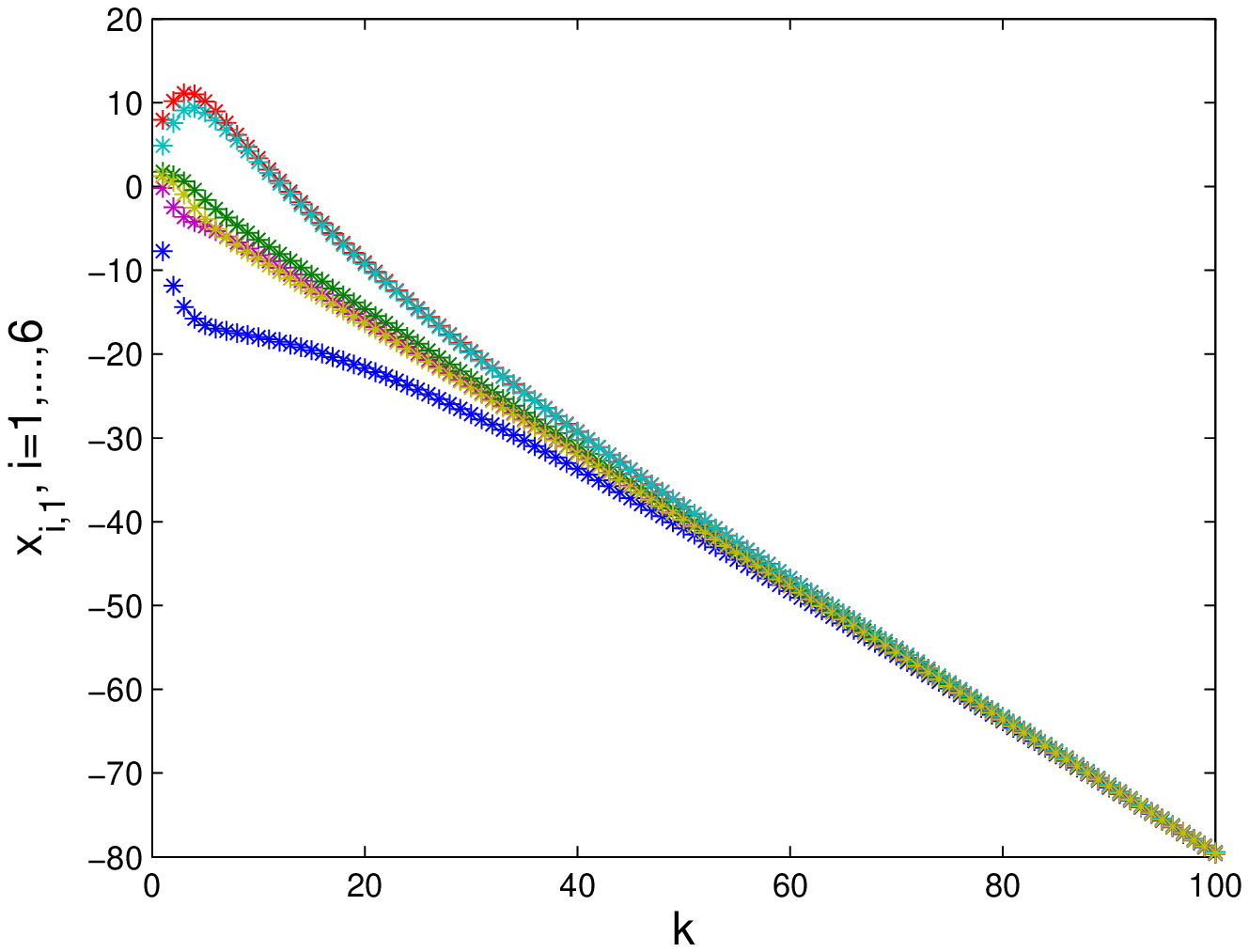}
\\
\includegraphics[width=0.8\linewidth,height=0.3\linewidth]{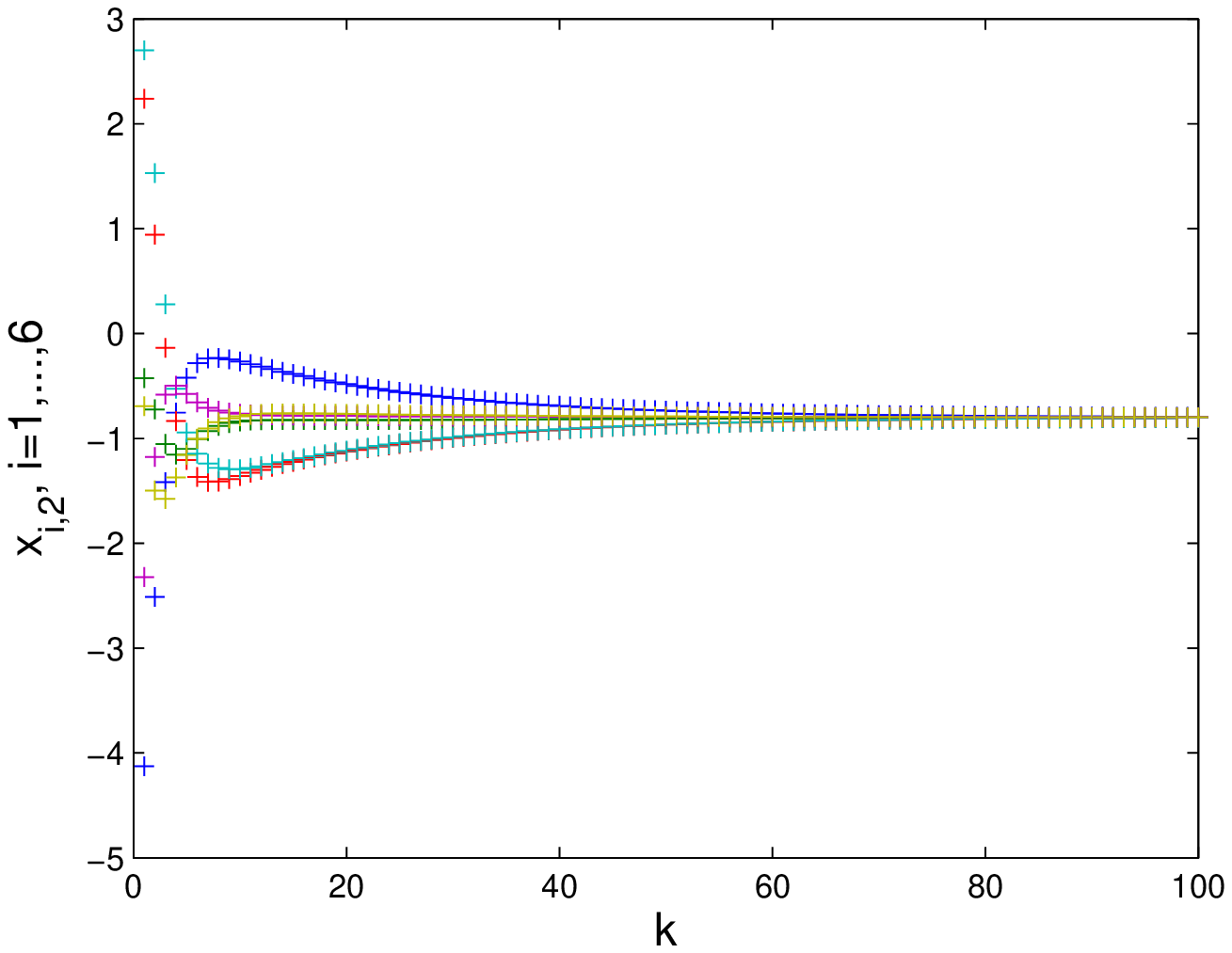}
\\
\includegraphics[width=0.8\linewidth,height=0.3\linewidth]{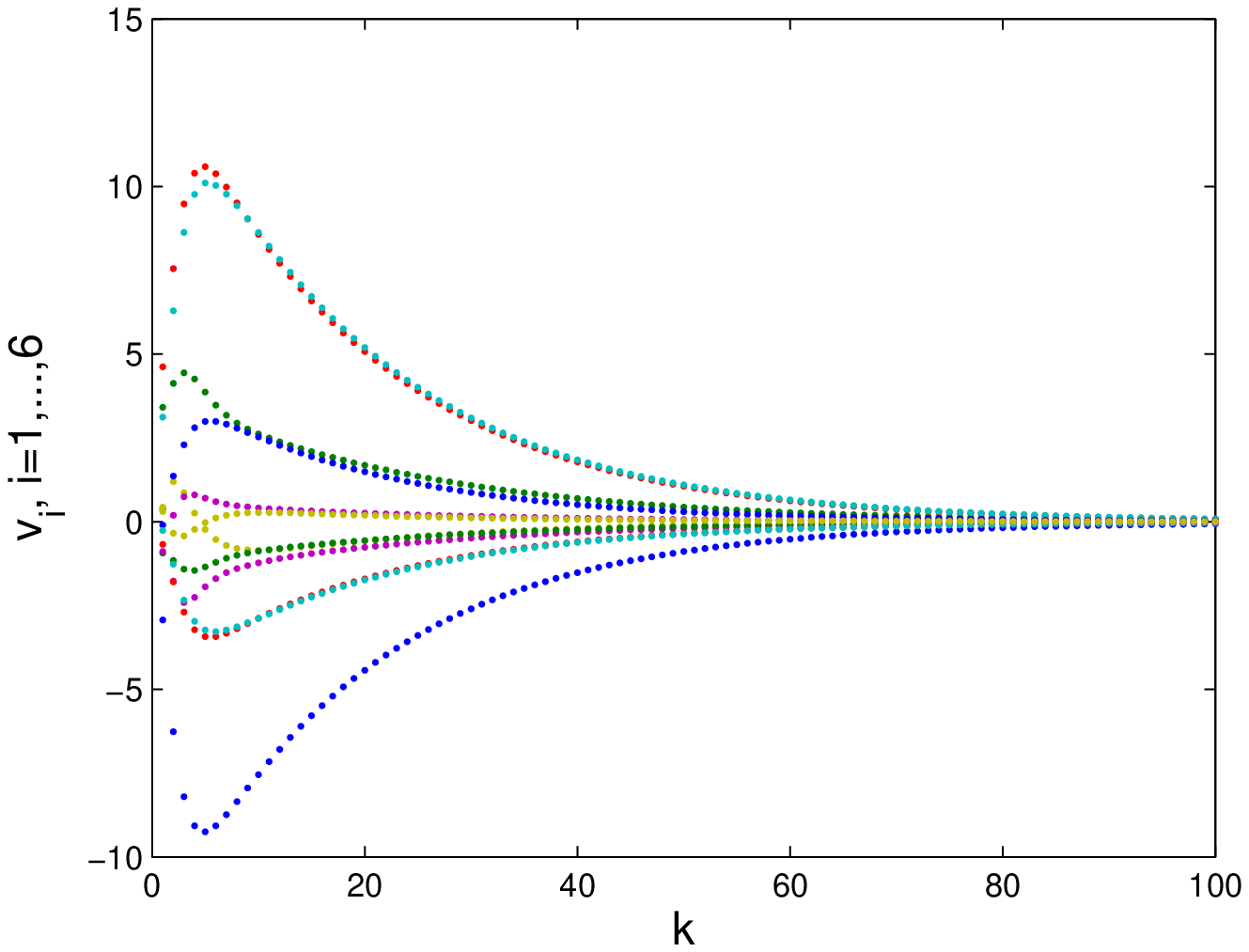}
\caption{The state trajectories of a double-integrator network.}
\end{center}
\end{figure}

\section{Application to Formation Control}

In this section, the consensus algorithms are modified to solve
formation control problems of multi-agent systems.

Let $\widetilde{H}=(h_1,h_2,\cdots,h_N)\in\mathbf{R}^{n\times N}$
describe a constant formation structure of the agent network in a
reference coordinate frame, where $h_i\in\mathbf{R}^{n}$, is the
formation variable corresponding to agent $i$. For example,
$h_1=\left[\begin{matrix}0&0\end{matrix}\right]^T$,
$h_2=\left[\begin{matrix}0&1\end{matrix}\right]^T$,
$h_3=\left[\begin{matrix}1&0\end{matrix}\right]^T$, and
$h_4=\left[\begin{matrix}1&1\end{matrix}\right]^T$ represent a unit
square. Variable $h_i-h_j$ denotes the relative formation vector
between agents $i$ and $j$, which is independent of the reference
coordinate.

A distributed formation protocol is proposed
as\begin{equation}\label{clf}
\begin{aligned}
v_i^+ &=(A+BK)v_i+L\left(\sum_{j=1}^Nd_{ij}C(v_i-v_j)-\tilde{\zeta}_i\right),\\
u_i &=Kv_i,
\end{aligned}
\end{equation}
where $ \tilde{\zeta}_i=\sum_{j=1}^Nd_{ij}(y_i-y_j-C(h_i-h_j))$ and
the rest of variables are the same as in \dref{cl}. It should be
noted that \dref{clf} reduces to the consensus protocol \dref{cl},
when $h_i-h_j=0$, $\forall\,i,j=1,2,\cdots,N$.


\begin{definition} The agents \dref{1} under protocol \dref{clf} achieve a given formation
$\widetilde{H}=(h_1,h_2,\cdots,h_N)$ if
\begin{equation}\label{confor}
\|x_i(k)-h_i- x_j(k)+h_j\|\rightarrow0,~\text{as} ~k\rightarrow
\infty, ~\forall~i,j=1,2,\cdots,N.
\end{equation}
\end{definition}

\begin{theorem} \label{thf1}
        For any $\mathcal {G}\in\Gamma_N$, the
agents \dref{1} reach the formation $\widetilde{H}$ under protocol
\dref{clf} if all the matrices $A+BK$, $A+(1-\lambda_i)LC$,
$i=2,\cdots,N$, are Schur stable, and $(A-I)(h_i-h_j)=0$,
$\forall\,i,j=1,2,\cdots,N$, where $\lambda_i$, $i=2,\cdots,N$,
denote the eigenvalues of $\mathcal {D}$ located in the open unit
disk.
\end{theorem}

\begin{proof}
Let $e_{x_i}=x_i-h_i-x_1+h_1$ and $e_{v_i}=v_i-v_1$, $i=2,\cdots,N$.
Then, the agents \dref{1} can reach the formation $\widetilde{H}$ if
and only if $e_{x_i}(k)\rightarrow0$, as $k\rightarrow \infty$,
$\forall\,i=2,\cdots,N$. By invoking $(A-I)(h_i-h_j)=0$,
$i,j=1,2,\cdots,N$, it follows from \dref{1} and \dref{clf} that
$$
\begin{aligned}
e_{x_i}^+ &=Ae_{x_i}+BKe_{v_i},\\
e_{v_i}^+
&=(A+BK)e_{v_i}+LC\left(\sum_{j=2}^Nd_{ij}(e_{v_i}-e_{v_j})-\sum_{j=2}^Nd_{1j}e_{v_j}\right.\\
&\quad
-\left.\sum_{j=2}^Nd_{ij}(e_{x_i}-e_{x_j})+\sum_{j=2}^Nd_{1j}e_{x_j}\right),~i=2,\cdots,N.
\end{aligned}
$$
Let $e_i=[e_{x_i}^T,e_{v_i}^T]^T$, $i=2,\cdots,N$, and
$e=[e_{2}^T,\cdots,e_{N}^T]^T$. Then, one has
\begin{equation}\label{netgf}
\begin{aligned}
e^+=\left(I_{N-1}\otimes\mathcal {A}+(I_{N-1}-\mathcal
{D}_2+\mathbf{1}_{N-1}\alpha)\otimes\mathcal {H}\right)e,
\end{aligned}
\end{equation}
where matrices $\mathcal {A}$, $\mathcal {H}$ are defined in
\dref{netg}, and $$\begin{aligned}
\alpha &=\begin{bmatrix}d_{12}& d_{13}& \cdots& d_{1N}\end{bmatrix},\\
\mathcal {D}_2 &=\begin{bmatrix}d_{22} & d_{23} & \cdots&
d_{1N}\\d_{32} & d_{33} & \cdots & d_{3N} \\ \vdots & \vdots &
\ddots & \vdots\\ d_{N2} & d_{N3} & \cdots & d_{NN}\end{bmatrix}.
\end{aligned}$$
By the definition of matrix $\mathcal {D}$, one can obtain
\cite{ma2}
$$S^{-1}(I_N-\mathcal {D})S=\begin{bmatrix} 0 & \alpha \\ 0 &
I_{N-1}-\mathcal {D}_2+\mathbf{1}_{N-1}\alpha\end{bmatrix}$$ with
$S=\begin{bmatrix} 1 & 0 \\ \mathbf{1}_{N-1} &
I_{N-1}\end{bmatrix}$. Therefore, the nonzero eigenvalues of
$I_N-\mathcal {D}$ are all the eigenvalues of $I_{N-1}-\mathcal
{D}_2+\mathbf{1}_{N-1}\alpha$. By following similar steps in the
Proof of Theorem \ref{th1}, one gets that system \dref{netgf} is
asymptotically stable if and only if all the matrices $A+BK$,
$A+(1-\lambda_i)LC$, $i=2,\cdots,N$, are Schur stable. This
completes the proof.
\end{proof}

\begin{remark}
Note that all kinds of formation structure can not be achieved for
the agents \dref{1} by using protocol \dref{clf}. The achievable
formation structures have to satisfy the constraints
$(A-I)(h_i-h_j)=0$, $\forall\,i,j=1,2,\cdots,N$. The formation
protocol \dref{clf} for a given achievable formation structure can
be constructed by using Algorithms 1 and 2. Theorem \ref{thf1}
generalizes the results given in \cite{laf, ren082}, where the agent
dynamics in \cite{laf, ren082} are restricted to be (generalized)
second-order integrators.
\end{remark}

\noindent{\bf Example 4}. Consider a network of $6$ double
integrators, described by
$$\begin{aligned}
x_i^{+} &=x_i+\tilde{v}_i, \\
\tilde{v}_i^{+} &=\tilde{v}_i+u_i,\\
y_i &=x_i,\quad i=1,2,\cdots,6,
\end{aligned}$$
where $x_i\in\mathbf{R}^2$, $\tilde{v}_i\in\mathbf{R}^2$,
$y_i\in\mathbf{R}^2$, and $u_i\in\mathbf{R}^2$ are the position, the
velocity, the measured output, and the acceleration input of agent
$i$, respectively.

The objective is to design a protocol \dref{clf} such that the
agents will evolve to a regular hexagon with edge length $8$. In
this case, choose $h_1=\left[\begin{matrix} 0 & 0 & 0 &
0\end{matrix}\right]^T$, $h_2=\left[\begin{matrix} 8 & 0 & 0 &
0\end{matrix}\right]^T$, $h_3=\left[\begin{matrix} 12 & 4\sqrt{3} &
0 & 0
\end{matrix}\right]^T$, $h_4=\left[\begin{matrix}
8 & 8\sqrt{3} & 0 & 0
\end{matrix}\right]^T$, $h_5=\left[\begin{matrix}
0 & 8\sqrt{3} & 0 & 0
\end{matrix}\right]^T$, $h_6=\left[\begin{matrix}
-4 & 4\sqrt{3} & 0 & 0
\end{matrix}\right]^T$.
As in Example 3, take $K=\left[\begin{matrix}-0.5I_2 \end{matrix}\right.\\
\left.\begin{matrix}-1.5I_2
\end{matrix}\right]$ and $L=\left[\begin{matrix}-1.051I_2&
-0.051I_2\end{matrix}\right]^T$ in protocol \dref{clf}. Then, the
agents with such a protocol \dref{clf} will form a regular hexagon
with respect to $\Gamma_{\leq 0.95}$. The sate trajectories of the
$6$ agents are depicted in Figure 4 for the communication topology
given in Figure 2.

\begin{figure}[htp]
\begin{center}
  \includegraphics[width=3.5in]{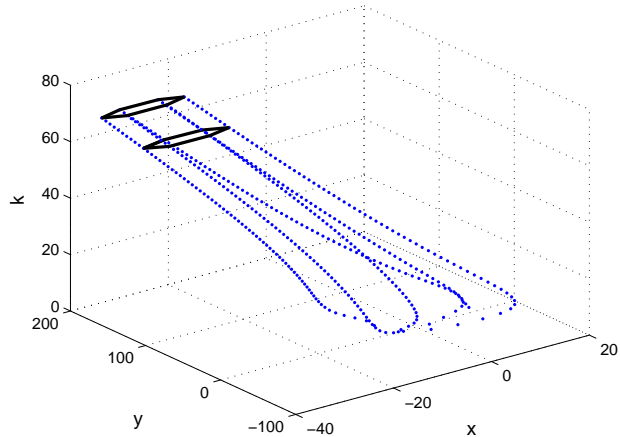}
  \caption{The agents form a regular hexagon.}
  \end{center}
\end{figure}

\section{Conclusions}

This paper has studied the consensus of discrete-time multi-agent
systems with linear or linearized dynamics. An observer-type
protocol based on the relative outputs of neighboring agents has
been proposed, which can be seen as an extension of the traditional
observer-based controller for a single system to one for multi-agent
systems. The consensus of high-dimensional multi-agent systems with
directed communication topologies can be converted into the
stability of a set of matrices with the same low dimension as that
of a single agent. The notion of discrete-time consensus region has
been introduced and analyzed. For neurally stable agents, an
algorithm has been presented to construct a protocol having a
bounded consensus region in the form of the open unit disk.
Moreover, for unstable agents, 
another algorithm has also been proposed to construct a protocol
having an origin-centered disk of radius $\delta$ ($0<\delta<1$) as
its consensus region. 
The consensus algorithms have been further
applied to solve formation control problems of multi-agent systems.
To some extent, this paper generalizes some existing results
reported in the literature, and opens up a new line for further
research on discrete-time multi-agent systems.




\medskip
Received xxxx 20xx; revised xxxx 20xx.
\medskip

\end{document}